\documentclass[article,notitlepage,amsthm,amsmath,amssymb]{revtex4-1}  
\usepackage{color}
\usepackage[dvipsnames]{xcolor}
\usepackage{graphicx}
\usepackage{subcaption}
\usepackage{amsmath}
\usepackage{amssymb}
\usepackage{amsthm}
\usepackage[colorlinks = true,linkcolor = blue,urlcolor = blue,citecolor = blue,anchorcolor = blue]{hyperref}

\usepackage{xy}
\xyoption{matrix}
\xyoption{frame}
\xyoption{arrow}
\xyoption{arc}

\usepackage{ifpdf}
\ifpdf
\else
\PackageWarningNoLine{Qcircuit}{Qcircuit is loading in Postscript mode.  The Xy-pic options ps and dvips will be loaded.  If you wish to use other Postscript drivers for Xy-pic, you must modify the code in Qcircuit.tex}
\xyoption{ps}
\xyoption{dvips}
\fi

\entrymodifiers={!C\entrybox}

\newcommand{\ket}[1]{{\left\vert{#1}\right\rangle}}
\newcommand{\qw}[1][-1]{\ar @{-} [0,#1]}
\newcommand{\qwx}[1][-1]{\ar @{-} [#1,0]}

\newcommand{\gate}[1]{*+<.6em>{#1} \POS ="i","i"+UR;"i"+UL **\dir{-};"i"+DL **\dir{-};"i"+DR **\dir{-};"i"+UR **\dir{-},"i" \qw}

\newcommand{\control}{*!<0em,.025em>-=-<.2em>{\bullet}}
\newcommand{\controlo}{*+<.01em>{\xy -<.095em>*\xycircle<.19em>{} \endxy}}
\newcommand{\ctrl}[1]{\control \qwx[#1] \qw}
\newcommand{\ctrlo}[1]{\controlo \qwx[#1] \qw}
\newcommand{\targ}{*+<.02em,.02em>{\xy ="i","i"-<.39em,0em>;"i"+<.39em,0em> **\dir{-}, "i"-<0em,.39em>;"i"+<0em,.39em> **\dir{-},"i"*\xycircle<.4em>{} \endxy} \qw}
\newcommand{\qswap}{*=<0em>{\times} \qw}
\newcommand{\multigate}[2]{*+<1em,.9em>{\hphantom{#2}} \POS [0,0]="i",[0,0].[#1,0]="e",!C *{#2},"e"+UR;"e"+UL **\dir{-};"e"+DL **\dir{-};"e"+DR **\dir{-};"e"+UR **\dir{-},"i" \qw}
\newcommand{\ghost}[1]{*+<1em,.9em>{\hphantom{#1}} \qw}

\newcommand{\lstick}[1]{*!R!<.5em,0em>=<0em>{#1}}

\newcommand{\Qcircuit}{\xymatrix @*=<0em>}

\definecolor{MyBlue}{rgb}{0.25,0.5,0.75}
\colorlet{c1}{MyBlue!90}
\colorlet{c2}{MyBlue!60}
\colorlet{c3}{MyBlue!30}

\definecolor{MyOrange}{rgb}{0.75,0.5,0.25}
\colorlet{c4}{MyOrange!30}
\colorlet{c5}{MyOrange!60}
\colorlet{c6}{MyOrange!90}

\usepackage{float}
\usepackage{newfloat,algcompatible}
\usepackage[size=small]{caption}
\usepackage{etoolbox}

\AtEndEnvironment{algorithm}{\noindent\hrulefill\par\nobreak\vskip-8pt}

\DeclareFloatingEnvironment[
    fileext=loa,
    listname=List of Algorithms,
    name=ALGORITHM,
    placement=tbhp,
]{algorithm}
\DeclareCaptionFormat{algorithms}{\vskip-15pt\hrulefill\par#1#2#3\vskip-6pt\hrulefill}
\algblock[Input]{Input}{EndInput}
\algblockdefx[Input]{Input}{EndInput}%
    [1]{\textbf{Input} #1}%
    {}

\begin{document}

\newcommand{\willers}[1]{\textbf{\color{blue}#1}}
\newcommand{\dmitri}[1]{\textbf{\color{Forestc2}#1}}

\newcommand{\Gate}[1]{\textsc{#1}}
\newcommand{\idgate}{\Gate{ID}}
\newcommand{\cnotgate}{\Gate{CNOT}} 
\newcommand{\czgate}{\Gate{CZ}}
\newcommand{\rzgate}{\Gate{rz}} 
\newcommand{\swapgate}{\Gate{SWAP}} 
\newcommand{\hgate}{\Gate{H}}
\newcommand{\xgate}{\Gate{X}}
\newcommand{\ygate}{\Gate{Y}}
\newcommand{\zgate}{\Gate{Z}}
\newcommand{\pgate}{\Gate{P}} 
\newcommand{\pdgate}{\Gate{Pd}}
\newcommand{\tgate}{\Gate{T}} 
\newcommand{\tdgate}{\Gate{Td}} 
\newcommand{\xxgate}{\Gate{xx}} 
\newcommand{\rgate}{\Gate{r}} 

\newtheorem{dfn}{Definition}
\newtheorem{prop}{Proposition}
\newtheorem{lemma}{Lemma}
\newtheorem{fact}{Fact}
\newtheorem{conj}{Conjecture}
\newtheorem{theorem}{Theorem}
\newtheorem{cor}{Corollary}

\newcommand{\eq}[1]{Eq.~(\ref{eq:#1})}
\renewcommand{\sec}[1]{\hyperref[sec:#1]{Section~\ref*{sec:#1}}}
\newcommand{\ssec}[1]{\hyperref[ssec:#1]{Subsection~\ref*{ssec:#1}}}
\newcommand{\fig}[1]{\hyperref[fig:#1]{Figure~\ref*{fig:#1}}}
\newcommand{\tab}[1]{\hyperref[tab:#1]{Table~\ref*{tab:#1}}}
\newcommand{\lem}[1]{\hyperref[lem:#1]{Lemma~\ref*{lem:#1}}}
\newcommand{\propos}[1]{\hyperref[propos:#1]{Proposition~\ref*{propos:#1}}}
\newcommand{\thm}[1]{\hyperref[thm:#1]{Theorem~\ref*{thm:#1}}}
\newcommand{\alg}[1]{\hyperref[alg:#1]{Algorithm~\ref*{alg:#1}}}

\title{CNOT circuits need little help to implement arbitrary Hadamard-free Clifford transformations they generate}

\author{Dmitri Maslov}
\affiliation{IBM Quantum, IBM T. J. Watson Research Center, Yorktown Heights, NY 10598, USA}
\author{Willers Yang}
\affiliation{IBM Quantum, IBM T. J. Watson Research Center, Yorktown Heights, NY 10598, USA}

\begin{abstract}
A Hadamard-free Clifford transformation is a circuit composed of quantum Phase ($\pgate$), $\czgate$, and $\cnotgate$ gates. It is known that such a circuit can be written as a three-stage computation, -P-CZ-CNOT-, where each stage consists only of gates of the specified type.

In this paper, we focus on the minimization of circuit depth by entangling gates, corresponding to the important time-to-solution metric and the reduction of noise due to decoherence. We consider two popular connectivity maps: Linear Nearest Neighbor (LNN) and all-to-all. First, we show that a Hadamard-free Clifford operation can be implemented over LNN in depth $5n$, i.e., in the same depth as the -CNOT- stage alone. This allows us to implement arbitrary Clifford transformation over LNN in depth no more than $7n{-}4$, improving the best previous upper bound of $9n$. Second, we report heuristic evidence that on average a random uniformly distributed Hadamard-free Clifford transformation over $n{>}6$ qubits can be implemented with only a tiny additive overhead over all-to-all connected architecture compared to the best-known depth-optimized implementation of the -CNOT- stage alone.  This suggests the reduction of the depth of Clifford circuits from $2n\,{+}\,O(\log^2(n))$ to $1.5n\,{+}\,O(\log^2(n))$ over unrestricted architectures.
\end{abstract}

\maketitle

\section{Introduction}\label{sec:Intro}

Clifford circuits are ubiquitous in quantum computing.  They appeared early in the study of quantum computations.  For example, Bernstein–Vazirani algorithm \cite{bernstein1993quantum}, which is a Clifford circuit in its entirety, served to establish a separation between quantum and classical computations in the black-box model.  In early (as well as modern) experimental work, randomized benchmarking \cite{knill2008randomized, magesan2011scalable}, performed entirely by Clifford circuits, was found to be a powerful tool for establishing the quality of quantum gates/computations.  Quantum error correction is based on Clifford circuits, including both logical state encoding (as well as the underlying formalism) \cite{gottesman1997stabilizer, nielsen2002quantum} and state purification required to obtain non-Clifford gates on the logical level \cite{bravyi2005universal}.  Fault-tolerant and even physical-level circuits (although slightly less so) are frequently considered over Clifford+$\tgate$ and/or Clifford+$R_z$ libraries---the role of the Clifford circuits follows from the library names themselves.  Other notable use cases include shadow tomography \cite{aaronson2018shadow}, study of entanglement \cite{bennett1996mixed}, and more. 

There are two popular choices to evaluate a circuit cost by a single number: weighted gate count, and weighted depth.  The weights depend on whether one works with physical-level or logical-level circuits and the details of physical-level control and error correcting approach, correspondingly.  Given that logical-level Clifford circuits can often be implemented transversely, the costs are most likely determined by the physical-level constraints.  Both leading quantum technologies, superconducting circuits \cite{superconducting} and trapped ions \cite{trappedions}, offer single-qubit gates at a fraction of the cost of the two-qubit gates.  It thus makes the most sense to discard the single-qubit gates and simply count the entangling gates.  We furthermore choose to focus on optimizing the depth by entangling gates rather than the gate count.  We believe this to be a better choice since the depth most accurately reflects the time to solution, and the purpose of any computation is above all to optimize this figure of merit. Secondly, minimizing depth helps to optimize the fidelity in computations where noise is dominated by decoherence, such as the case in current superconducting circuit devices \cite{superconducting}.  For quantum technologies where noise may be dominated by a coherent source, such as the trapped ions \cite{linke2017experimental}, one may too prefer to optimize depth, since the noise due to random over/underrotations may scale sublinearly with the number of gates \cite{linke2017experimental} or even as slow as the square root \cite{campbell2017shorter}, whereas the effect of the decoherence is always linear with time.  Thus, a smaller decoherence noise channel may cause a larger disturbance in sufficiently deep computations compared to a coherent over/underrotation channel, due to coherent error cancellations.

All quantum computations need to be broken down into single-qubit and two-qubit gates, which are then implemented in hardware.  A single-qubit gate can be implemented by directly addressing the respective physical qubit, but the two-qubit gates require addressing a pair of qubits.  Depending on the technology and particulars of the controlling apparatus, it may not always be possible.  Of the two leading technologies, trapped ions \cite{trappedions} is generally said to offer all-to-all qubit connectivity.  However, this does not remain the case for arbitrary numbers of qubits, and it is possible that keeping all-to-all connectivity beyond 100 qubits can be challenging \cite{linke2017experimental}.  In our work, we address the possibility of unrestricted architectures by considering the task of depth optimization without limiting the pairs of qubits to which an elementary two-qubit gate can apply.  In contrast, superconducting circuits technology \cite{superconducting} offered limited connectivity at the offset, including the first 5-qubit IBM quantum processor that premiered in 2016 \cite{ibm2016ibm}.  Current offerings by the companies focusing on superconducting technology include a 2-dimensional square lattice (Google), heavy hexagonal lattice (where each edge of a hexagonal lattice contains a qubit in the middle, IBM), 2-dimensional square lattice with diagonal terms (IBM, presently retired Tokyo architecture), and a square lattice with octagons (Rigetti).  What unites these is the ability to find a chain subgraph in each, that is, embed the Linear Nearest Neighbor (LNN) architecture.  A chain is frequently a subgraph of many other graphs and lattices.  As a result, and due to its universality, we also focus on implementing Clifford circuits over the LNN architecture. 

Clifford circuits have been studied quite extensively, leading to numerous efficient implementations addressing a number of metrics of value.  A first noteworthy result in this area was the development of an 11-stage layered decomposition -H-CNOT-P-CNOT-P-CNOT-H-P-CNOT-P-CNOT- \cite{aaronson2004improved}, which, together with gate-count asymptotically optimal synthesis algorithm for linear reversible circuits \cite{patel2008optimal} settled the asymptotic gate complexity.  This asymptotically optimal algorithm can be parallelized \cite{de2021reducing, jiang2020optimal} to offer depth upper bound guarantee of $O(n/\log(n))$, leading to asymptotic depth optimality.  However, the leading constant in front of the $n/\log(n)$ term is as high as 20 \cite{de2021reducing}, resulting in circuits that become advantageous compared to those featuring a higher asymptotic complexity of $O(n)$ only when the number of qubits exceeds $1{,}345{,}000$ or more \cite{de2021reducing, maslov2022depth}.  Indeed, \cite{maslov2022depth} offers an improvement over \cite{de2021reducing}, implying that the performance crossover number $1{,}345{,}000$ was likely pushed higher by the newer work.  In either case, we do not anticipate executing Clifford circuits spanning this many qubits on quantum computers.  This suggests that the currently known constructions achieving asymptotic optimality offer theoretical guarantees but are not necessarily practical.  We also note that constructions offering asymptotic optimality and exploring ancillary space have also been considered \cite{jiang2020optimal}.  However, leading constants are high \cite{de2021reducing}.   

A more practical approach to implementing Clifford circuits would be to first start with a short layered decomposition -X-Z-P-CNOT-CZ-H-CZ-H-P- \cite{duncan2020graph, bravyi2021hadamard}, offering only three entangling stages, as opposed to the original five \cite{aaronson2004improved}, and noting that a -CZ- stage is `simpler' than a -CNOT- stage.  Then, depending on the qubit connectivity:
\begin{itemize}
    \item over LNN: -X-Z-P-CNOT-CZ-H-CZ-H-P- can be rewritten as -X-P-CNOT$\widehat{\text{-CZ-}}$H$\widehat{\text{-CZ-}}$H-P-, where $\widehat{\text{-CZ-}}$ is a -CZ- stage together with qubit reversal. The -CNOT- stage can be implemented in depth at most $5n$ \cite{kutin2007computation}, and $\widehat{\text{-CZ-}}$ can be implemented in depth $2n{+}2$ \cite{maslov2018shorter}, which combined with some local optimizations leads to the depth $9n$ implementation of arbitrary Clifford circuits;
    \item over all-to-all: both -CZ- and -CNOT- stages are best implemented according to \cite{maslov2022depth}, resulting in depth $2n+O(\log^2(n))$ circuits.
\end{itemize}
Here, we show that a -CNOT-CZ- stage can always be implemented in depth $5n$ over LNN, leading to Clifford circuit depth guarantee of no more than $7n{-}4$.  We also show computational evidence that a -CNOT-CZ- stage can be implemented in depth $n+O(\log^2(n))$ over all-to-all architecture, leading to a likely upper bound of $1.5n+O(\log^2(n))$ on the depth of Clifford circuits in unrestricted architectures.

The rest of the paper is organized as follows. \sec{background} discusses basic background on the subject of $\czgate$ and $\cnotgate$ circuits, highlighting how sets of linear functions can form a basis in the linear space of $\czgate$ circuits. \sec{lnn} focuses on developing circuits over LNN architecture.  It describes a depth-$5n$ implementation of the joint -CNOT-CZ- stage and offers an $O(n^3)$ algorithm to construct it.  We employ this construction to implement arbitrary Clifford operation in depth $7n{-}4$. \sec{all-to-all} reports a heuristic solution to the problem of optimizing circuit depth of the implementation of the -CNOT-CZ- stage in the all-to-all connected architecture.  Final remarks and conclusion can be found in \sec{conc}.

\section{Background}\label{sec:background}
The core of technical discussions is the study of quantum circuits composed with Phase ($\pgate$), $\czgate$, and $\cnotgate$ gates.  These gates can be defined as the unitary matrices or by the transformations they perform over basis states, as follows: 
\begin{align*}
\pgate(x): \ket{x}&\mapsto i^x\ket{x},\\ 
\czgate(x,y): \ket{x,y}&\mapsto (-1)^{xy}\ket{x,y},\text{ and} \\
\cnotgate(x,y): \ket{x,y}&\mapsto \ket{x,x\oplus y}. 
\end{align*}
Such circuits compute Hadamard-free Clifford transformations and form a finite group.  It was shown \cite{maslov2018shorter} that the circuits with Phase, $\czgate$, and $\cnotgate$ gates can be computed as a three-stage layered computation -P-CZ-CNOT-, where each layer can be composed using the gates of its specified type, and the layers can come in any order.  This layered decomposition offers an efficient way of implementing Hadamard-free circuits. Indeed, -P- layer may consist of the gates $\idgate$ (identity), $\pgate$, $\zgate\,{=}\,\pgate^2$, and $\pgate^\dagger\,{=}\,\pgate^3$ on each qubit, and thus is trivial to compose.  A -CZ- layer can be implemented straightforwardly and optimally by the $\czgate$ gates alone (see next paragraph).  A -CNOT- layer, also known as linear reversible circuits, has been studied extensively \cite{moore2001parallel, kutin2007computation, patel2008optimal, jiang2020optimal, maslov2022depth}.  We note that both -CZ- layer can be implemented more efficiently by utilizing $\pgate$ and $\cnotgate$ gates \cite{maslov2018shorter} and linear reversible circuits can be improved by utilizing $\pgate$ and $\hgate$ (Hadamard) gates \cite{bravyi2021hadamard}.  Here, we employ related rules to reduce the depth of the decomposition -P-CZ-CNOT- by mixing the individual layers.
 
First, recall some basic properties of circuits composed with $\czgate$ gates.  $\czgate(x,y)\,{=}\,\czgate(y,x)$, meaning the order of control and target does not matter.  A $\czgate$ gate is self-inverse; in other words, two neighboring $\czgate$ gates operating over the same set of qubits can be removed from the circuit without affecting its functionality.  Any two $\czgate$ gates commute.  This means that any pair of identical $\czgate$ gates can be removed and allows to express arbitrary $\czgate$ circuit over $n$ qubits by listing a set of at most $\frac{(n-1)n}{2}$ pairs of qubits to which such gates apply.  We further conclude that $\czgate$ circuits span the linear space $\mathbb{F}_2^{(n-1)n/2}$ over the binary field $\mathbb{F}_2\,{:=}\,\{0,1\}$.  A linear space has a basis, which will become an important consideration in future discussions.  The standard basis is $\{\czgate(i,j)|_{1 \leq i < j \leq n}\}$.

A $\czgate$ circuit can also be implemented by inserting Phase gates into linear reversible circuits, by considering phase polynomials. To illustrate how this works, consider a well-known circuit identity:
\begin{equation*}
\Qcircuit @C=0.7em @R=1.7em @!{
\lstick{\ket{x}} & \ctrl{1}  & \qw \\ 
\lstick{\ket{y}} & \ctrl{-1} & \qw
}
\hspace{1em}\raisebox{-0.9em}{=}\hspace{1em}
\Qcircuit @C=0.2em @R=0.2em @!{
& \gate{\pgate}   & \ctrl{1} & \qw                  & \ctrl{1} 	& \qw \\ 
& \gate{\pgate}   & \targ    & \gate{\pgate^\dagger}& \targ     & \qw }
\end{equation*}
On the left hand side, we have a $\czgate(x,y)$ which applies the phase $-1^{x\cdot y}=i^{2x\cdot y}$; and on the right hand side, we have three Phase gates, $\pgate{\ket{x}},\pgate{\ket{y}},$ and $\pgate^\dagger{\ket{x\oplus y}}$, which apply phases $i^x$, $i^y$ and $i^{-x\oplus y}$. The equality holds due to the mixed arithmetic identity $2x{\cdot} y \equiv x + y - x{\oplus} y \ (\bmod \,\,4)$.  Since the primary variables, $x$ and $y$, are available to experience the application of Phase gates on the input side of the circuit, it shows that the gate $\czgate(x,y)$ can be computed via Phase gate insertion by a linear reversible circuit that computes the linear function $x \oplus y$ at some point.  Furthermore, a linear reversible circuit spanning qubits $x_1,x_2,...,x_n$ and computing linear functions $x_i {\oplus} x_j$ for all $i,j: 1 \leq i < j \leq n$ offers enough opportunity to insert Phase gates to compute arbitrary $\czgate$ gate transformation, since $\{\czgate(i,j)|_{1 \leq i < j \leq n}\}$ is a basis.

Next consider a linear function $f = x_{i_1}\oplus x_{i_2}\oplus... \oplus x_{i_k}$, spanning a subset of $k$ qubits.  It is known and is easy to verify explicitly that up to a layer of Phase gates applied to primary variables on the input side of the circuit, a Phase gate applied to $f$ is equivalent to the CZ circuit $\prod_{1\leq s<t\leq k} \czgate(x_{s},x_{t})$.  This observation offers a natural way to treat arbitrary linear functions as $\czgate$ circuits or otherwise vectors in the $\frac{(n-1)n}{2}$-dimensional linear space spanning $\czgate$ circuits. 

Finally, we recall another important identity, that we express using phase polynomials as follows:  
\begin{equation}\label{eq:7Term}
a + b + c - a{\oplus}b - a{\oplus}c - b{\oplus}c + a{\oplus}b{\oplus}c  \equiv 0 (\bmod 4).
\end{equation} 
Here, each literal corresponds to some linear function of primary variables. A summand with the positive sign expresses the application of the $\pgate$ gate to it, a summand with the negative sign indicates the application of the $\pgate^\dagger$ gate, and the sum is taken modulo-4 to reflect the phase identity $i^4\,{=}\,1$. This seven-term identity says that a circuit implementing the set of seven Phase operations participating as summands computes the identity function. For the purpose of our work, \eq{7Term} is a very important identity that establishes linear dependence of linear functions as vectors in the space $\mathbb{F}_2^{(n-1)n/2}$---not to mistake with the linear dependence as regular linear functions.  Applying this identity most frequently comes in the form of rewriting one of seven terms as a combination of the remaining six; it lies at the core of the constructions that follow. In particular, we focus on the discovery of a full $\czgate$ basis in the linear functions computed inside a -CNOT- stage as means for merging a -CZ- stage with the -CNOT- stage.

\section{-CNOT-CZ- circuits over Linear Nearest Neighbour architecture} \label{sec:lnn}

First, consider the LNN architecture, where qubits are arranged in a line and only adjacent qubits are allowed to interact. Our main result can be summarized as follows.
\begin{theorem}\label{thm:1}
Any Hadamard-free Clifford transformation over $n$ qubits can be implemented over LNN architecture as a circuit with the two-qubit gate depth of at most $5n$. 
\end{theorem}

At first glance, \thm{1} may be surprising. In the LNN architecture, a -CNOT- layer alone requires a circuit of depth $5n$ when constructed using the best known synthesis algorithm \cite{kutin2007computation}. Since Hadamard-free Clifford transformations can be written as a three stage computation -P-CZ-CNOT- \cite{maslov2018shorter}, our result implies that the -CZ- layer adjacent to a -CNOT- layer can be implemented ``for free" (without increasing the two-qubit gate depth) over the LNN architecture through $\pgate$ gate insertions. As we'll show in \sec{cliff}, \thm{1} also allows us to implement any Clifford operation as a circuit of depth $7n{-}4$ over LNN, improving the prior best-known upper bound of $9n$ \cite{bravyi2021hadamard}. In the remainder of this section, we take a closer look at the -CNOT- synthesis algorithm of \cite{kutin2007computation} in \ssec{CNOT} before describing an efficient Hadamard-free synthesis algorithm (\ssec{CZ}) that achieves the advertised upper bound in \thm{1}.

\subsection{-CNOT- Synthesis over LNN}\label{ssec:CNOT}

The best-known synthesis algorithm for $\cnotgate$ circuit synthesis over LNN achieves depth $5n$ by starting with two back-to-back copies of the sorting network, $C_1$ and $C_2$, and modifying them according to a set of rules \cite{kutin2007computation}. Since a $\cnotgate$ transformation corresponds to a linear reversible operation, we can represent it as an $n{\times}n$ invertible Boolean matrix $M$. Application of the $\cnotgate$ gate performs modulo-2 addition (EXOR, denoted $\oplus$) of the matrix column corresponding to the control qubit to the column corresponding to the target qubit of the given $\cnotgate$ gate. The synthesis of $M$ using $\cnotgate$ gates can therefore be thought of as the diagonalization of $M$ via column operations. The diagonalization is done in two steps \cite{kutin2007computation}: first, $M$ is transformed into a northwest-triangular matrix $M'$ using a modification of the $\cnotgate$ circuit $C_1$ of depth at most $2n$, and then, $M'$ is diagonalized by modifying the sorting network $C_2$ in depth at most $3n$ \cite{kutin2007computation}.  An $n{\times}n$ matrix is called northwest-triangular if all entries below the antidiagonal are $0$. The combined circuit $C_1C_2$ diagonalizes $M$ and has depth at most $2n{+}3n \,{=}\, 5n$; therefore, $(C_1C_2)^{-1}$ is a depth-$5n$ circuit that implements $M$.  

We focus our attention on $C_2$, the circuit diagonalizing a northwest-triangular matrix, since it has a well-defined structure. In \ssec{CZ}, we prove that $C_2$ always generates enough linear functions to induce any -CZ- circuit via Phase gate insertions. $C_2$ starts as the sorting network \cite{kutin2007computation} and uses $(n{-}1)n/2$ ``boxes'' in $n$ alternating layers, each being either $\swapgate$ or $\swapgate^+$.  The $\swapgate^+$ gate is a two-qubit gate that acts on its inputs as follows, $\swapgate^+(x,y): \ket{x,y}\mapsto \ket{y, x\oplus y}$. In the circuit notation,
\[
\hspace{1em}\raisebox{-0.9em}{$\swapgate^+(x,y):=$}\hspace{3em}
\Qcircuit @C=0.7em @R=1.7em @!{
\lstick{\ket{x}} & \qswap \qwx[1]  & \qw &&\ket{y}\\ 
\lstick{\ket{y}} & \ctrlo{-1} & \qw&&\ket{x\oplus y}
}
\hspace{2em}\raisebox{-0.9em}{$=$}\hspace{3em}
\Qcircuit @C=0.7em @R=1.3em @!{
\lstick{\ket{x}} & \ctrl{1}  & \targ&\qw &&\ket{y}\\ 
\lstick{\ket{y}} & \targ & \ctrl{-1}& \qw&&\ket{x\oplus y}.
}
\]
Whether a box within the sorting network is a $\swapgate$ or $\swapgate^+$ is determined by the corresponding northwest-triangular matrix \cite{kutin2007computation}. In fact, the correspondence between northwest-triangular matrices and such sorting networks is bijective, which can be established by the counting argument.  An example of a sorting network on $6$ qubits and the northwest-triangular matrix it diagonalizes is given in \fig{C}.

\begin{figure}[h]

$$M =  \begin{bmatrix}
0 & 0 & 1 & 0 & 1 & 1 \\
0 & 0 & 0 & 1 & 1 & 0 \\
0 & 1 & 1 & 1 & 0 & 0 \\
1 & 1 & 1 & 0 & 0 & 0 \\
0 & 1 & 0 & 0 & 0 & 0 \\
1 & 0 & 0 & 0 & 0 & 0 
\end{bmatrix},$$
\[
\raisebox{-3em}{$C: $}\hspace{6em}
\Qcircuit @C=1.2em @R=1.2em @! { \\& \lstick{ x_4\oplus x_6:  } & \qswap & \qw & \qswap \qwx[1] & \qw & \qswap & \qw & \qw &{x_1  }\\
	 	 & \lstick{ x_3\oplus x_4\oplus x_5 :  } & \qswap \qwx[-1] & \qswap & \controlo\qw & \qswap & \qswap \qwx[-1] & \qswap \qwx[1] & \qw&{x_2  } \\
	 	 & \lstick{x_1\oplus x_3\oplus x_4:  } & \qswap \qwx[1] & \qswap \qwx[-1] & \qswap & \qswap \qwx[-1] & \qswap \qwx[1] & \controlo\qw & \qw &{x_3  }\\
	 	 & \lstick{ x_2\oplus x_3 :  } & \controlo\qw & \qswap \qwx[1] & \qswap \qwx[-1] & \qswap & \controlo\qw & \qswap \qwx[1] & \qw &{x_4  }\\
	 	 & \lstick{x_1\oplus x_2:  } & \qswap \qwx[1] & \controlo\qw & \qswap \qwx[1] & \qswap \qwx[-1] & \qswap \qwx[1] & \controlo\qw & \qw&{x_5  } \\
	 	& \lstick{ x_1:  } & \controlo\qw & \qw & \controlo\qw & \qw & \controlo\qw & \qw & \qw&{x_6  } \\
\\ }
\]

    \caption{Initially, qubit $i$ contains the linear function represented by the $i^{th}$ column of $M$. Circuit $C$ diagonalizes the matrix $M$.  Outputs of the circuit $C$ are the primary variables.}
    \label{fig:C}
\end{figure}

We next label qubits and boxes of the sorting network and define an ordering of the boxes. Using the labeling and ordering, we make several important observations on the structure of the sorting network that lie at the heart of our algorithm. 

First, label each qubit $i$ on the output side of the circuit with the label $n{+}1{-}i$ in the input side.  This is justified by the observation that the sorting network inverts the order of the qubits.  When a box (either $\swapgate$ or $\swapgate^+$) acts on the qubits $i$ and $i{+}1$ with labels $l_{i}$ and $l_{i+1}$, we swap the labels. As shown in \cite[Section 7.3]{kutin2007computation}, the label of a qubit always equals to the largest index of the primary variable stored in that qubit.  For example, in \fig{C}, the first qubit initially holds the value $x_4\oplus x_6$, corresponding to its label, $6\,{=}\,6{+}1{-}1$. It can also be shown that for each pair of labels $i{\neq}j$, there is exactly one box that changes $i$ and $j$, and  this box always shuffles the larger label downward and the smaller label upward. That is, whenever a box swaps two labels $(i,j)\mapsto(j,i)$, then ${i}>j$. It follows that we can uniquely label each box with two indices $i$ and $j$, where $i$ and $j$ are the labels of the qubits it swaps.  An example of this labeling is shown in \fig{Labeling}.

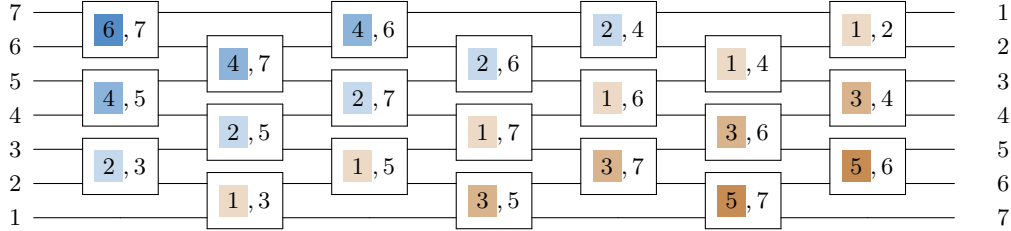
\begin{figure}
\[
\scalebox{1.0}{
\Qcircuit @C=2.0em @R=.5em @!R { \\
\lstick{{7 }    } &\multigate{1}{\colorbox{c1}{6},7} & \qw& \multigate{1}{\colorbox{c2}{4},6} & \qw&\multigate{1}{\colorbox{c3}{2},4} & \qw& \multigate{1}{\colorbox{c4}{1},2}  & \qw &1\\
\lstick{{6 }    } & \ghost{\colorbox{red}{6},7}& \multigate{1}{\colorbox{c2}{4},7} & \ghost{\colorbox{c2}{4},6}& \multigate{1}{\colorbox{c3}{2},6} & \ghost{\colorbox{c3}{2},4}&\multigate{1}{\colorbox{c4}{1},4} & \ghost{\colorbox{c4}{1},2}& \qw &2\\
\lstick{{5 }    } & \multigate{1}{\colorbox{c2}{4},5} & \ghost{\colorbox{c2}{4},7}&\multigate{1}{\colorbox{c3}{2},7} & \ghost{\colorbox{c3}{2},6}& \multigate{1}{\colorbox{c4}{1},6} & \ghost{\colorbox{c4}{1},4}& \multigate{1}{\colorbox{c5}{3},4}  & \qw &3\\
\lstick{{4 }    } & \ghost{\colorbox{c2}{4},5}& \multigate{1}{\colorbox{c3}{2},5} & \ghost{\colorbox{c3}{2},7}&\multigate{1}{\colorbox{c4}{1},7} & \ghost{\colorbox{c4}{1},6}&\multigate{1}{\colorbox{c5}{3},6} & \ghost{\colorbox{c5}{3},4}& \qw &4\\
\lstick{{3 }    } &\multigate{1}{\colorbox{c3}{2},3} & \ghost{\colorbox{c3}{2},4}& \multigate{1}{\colorbox{c4}{1},5} & \ghost{\colorbox{c4}{1},7}& \multigate{1}{\colorbox{c5}{3},7} & \ghost{\colorbox{c5}{3},6}& \multigate{1}{\colorbox{c6}{5},6} & \qw &5\\
\lstick{{2 }    } & \ghost{\colorbox{c3}{2},3}&\multigate{1}{\colorbox{c4}{1},3} & \ghost{\colorbox{c4}{1},5}& \multigate{1}{\colorbox{c5}{3},5} & \ghost{\colorbox{c5}{3},7}&\multigate{1}{\colorbox{c6}{5},7} & \ghost{\colorbox{c6}{5},6}& \qw &6\\
\lstick{{1 }    } & \qw & \ghost{\colorbox{c4}{1},3}& \qw & \ghost{\colorbox{c5}{3},5}& \qw & \ghost{\colorbox{c6}{5},7}& \qw & \qw &7\\
\\ }}
\]
\caption{Example labeling and layer ordering. For $n=7$, $P_n = \{6,4,2,1,3,5\}$, which is the index shared by the six diagonal layers from left to right, with each layer highlighted in a different color. }
    \label{fig:Labeling}
\end{figure}

Due to the above observations, boxes in the same bottom-left to top-right diagonal layers share a common first index, as illustrated in \fig{Labeling}.  We label each diagonal layer by the index they share, which gives rise to the following sequence,
\begin{displaymath}
P_n:=\left\{
\begin{array}{ll}
    (n{-}1,n{-}3,...,1,2,4,...,n{-}2) & \text{for even }n\\
    (n{-}1,n{-}3,...,2,1,3,...,n{-}2) & \text{for odd }n
\end{array} \right..
\end{displaymath}
This labeling is related to $P_j$ labeling in \cite{maslov2018shorter}.  For two distinct boxes, $box(i,j)$ and $box(k,l)$, where $i{<}j$ and $k{<}l$, we write $box(i,j)\,{\prec}\, box(k,l)$ when $i$ precedes $k$ in $P_n$, and $box(i,j) \,{=}\, box(k,l)$ when $i{=}k$.

Intuitively, boxes are ordered by the diagonal layers they belong to, from left to right. As the content of a qubit travels from left to right, it either stays in the same layer (when the qubit is shuffled up) or goes to the next layer (when the qubit is shuffled down or is left out of a round of shuffling). The layer ordering gives us an important guarantee that aids in proofs that follow:
\begin{prop}\label{propos:order}
For each $k$ and $i{<}j$, $box(j,k)\,{\prec}\, box(i,j)$ implies $box(i,k)\,{\preceq}\, box(i,j)$.
\end{prop}
\begin{proof}
We can show that the contrapositive is true: $box(i,k) \succ box(i,j)$ implies $box(i,k)$ is ordered by $k$ (i.e. $k{<}i{<}j$) and $i$ precedes $k$ in $P_n$ implies $box(k,j)$ is also ordered by $k$, $box(k,j)\succ box(i,j)$. 

The validity of the proposition may also be established intuitively by tracking the qubits that the qubit $x_j$ comes in contact with. At $box(i,j)$, which is located in layer $i$, a linear function containing $x_j$ comes in contact with a linear function containing $x_i$. Suppose $x_j$ was also brought in contact with $x_k$ through $box(j,k)$ before layer $i$. If $k{>}i$, then $box(i,k)$ is in the layer $i$. Otherwise, if $k{<}i{<}j$, then both $box(i,k)$ and $box(j,k)$ are in the layer $k$, which comes before layer $i$ since $box(j,k)$ comes before $box(i,j)$. Putting them together, we have that $x_k$ and $x_i$ meet either in the layer $i$ or in some layer before.
\end{proof}

\subsection{Inducing Arbitrary $\czgate$ Circuits} \label{ssec:CZ} 

We now have enough vocabulary to prove \thm{1}. An outline of the proof is as follows: we show that a northwest-triangular matrix diagonalization circuit generates a set of linear functions of primary variables that forms a basis in the CZ circuit space by induction on $k$, the number of $\swapgate^+$ boxes it contains.  Then, in \ssec{Algo} we describe an algorithm to find a schedule of Phase gates that induces a $\czgate$ circuit in $O(n^3)$ time, where $n$ is the number of qubits.

We focus on the northwest-triangular matrix diagonalization circuit $C$.
We first deal with the base case $k{=}0$.

\begin{lemma}\label{lem:base}
Let $C$ be a northwest-triangular diagonalization circuit consisting of only $\swapgate$ gates. Then, $C$ generates the basis $\{x_i\oplus x_j | 1{\leq}i{<}j{\leq} n\}$ in the space of $\czgate$ circuits.
\end{lemma}

We omit the proof, since it has been addressed in \sec{background} and overall is easy to obtain by noting that the $\swapgate$ generates EXOR of the input qubits when implemented by the $\cnotgate$ gates, and that in the sorting network each qubit $i$ meets every other qubit $j$ exactly once.

Assume $C$ has $k{>}0$ $\swapgate^+$ boxes. Consider the leftmost $\swapgate^+$ in $C$ in the smallest layer by its order.  Let us refer to this box as $box(i,j)$, where $i$ and $j$ are its labels, $i{<}j$. If we replace $box(i,j)$ with a $\swapgate$, we obtain another circuit $C'$ with the corresponding matrix $M'$ that has $k{-}1$ $\swapgate^+$ boxes.  By the induction hypothesis, $C'$ contains a $\czgate$ basis.  An example of $C$ and $C'$ is shown in \fig{induction}. 

\begin{figure}[h]
    \centering
    \includegraphics[width=\textwidth]{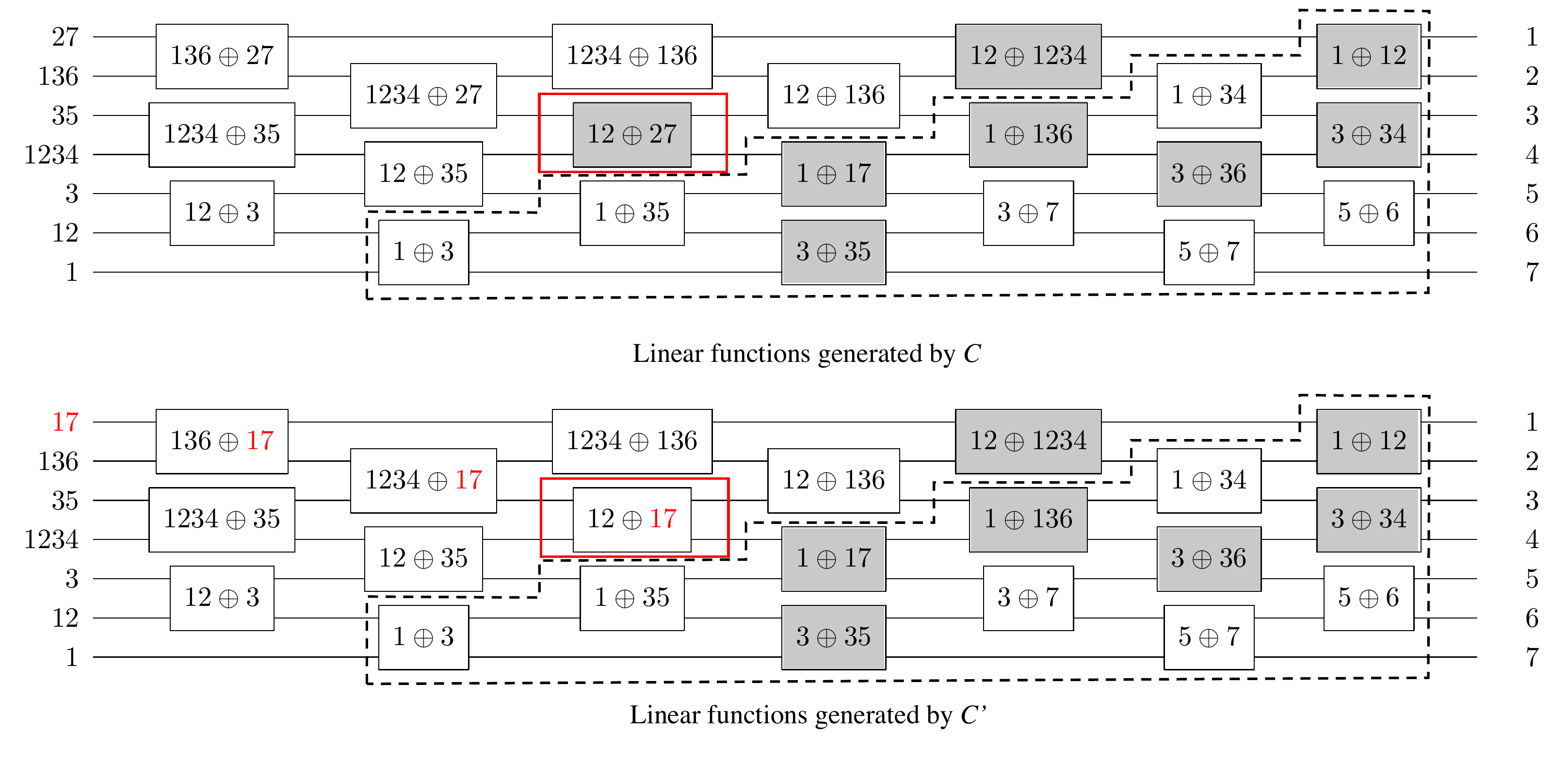}
    \caption{For simplicity, qubit concatenation refers to (output) qubit label EXORs. The shaded box represents a $\swapgate^+$ gate, whereas the unshaded box represents a $\swapgate$ gate.  Leftmost $\swapgate^+$s in the smallest order layer labeled by $(2,7)$ is enclosed by the red box---this is the one we focus on in the induction.  All linear functions generated by the circuit after $box(2,7)$ in higher layer order are enclosed by the dashed lines; they are identical in $C$ and $C'$.  All linear functions generated by boxes can be written as the EXORs of two inputs.  The differences between linear functions computed by $C$ and $C'$ are highlighted in red.}
    \label{fig:induction}
\end{figure}

We make two observations. First, the subcircuits following $box(i,j)$ compute identical linear functions between $C$ and $C'$. Also, all boxes $box(k,l)$ such that $box(k,l)\,{\succ}\, box(i,j)$ generate the same linear functions in both $C$ and $C'$ (for illustration, see dashed boxes in \fig{induction}). Second, the circuit preceding $box(i,j)$ consists of only $\swapgate$ gates, and as such it only swaps the linear functions without introducing new ones. For all boxes $box(k,l)\preceq box(i,j)$, $box(k,l)$ generates $c_l\oplus c_k$ in $C$ and $c'_l\oplus c'_k$ in $C'$. We know from the first observation that the linear function stored in each qubit after layer $i$ must be identical between $C$ and $C'$. Matching the outputs of the boxes on layer $i$, we can show that all but one initial input of $C$ and $C'$ are identical; that is, $c_k\,{=}\,c'_k$ for all $k\,{\neq}\, j$, and $c'_j = c_i\oplus c_j$ (or equivalently $c_j = c'_i\oplus c'_j$ ). We formalize these observations in the following Lemma:

\begin{lemma}\label{lem:dif}
Let $C$ be a modification of the sorting network with $k{>}0$ $\swapgate^+$ boxes, and let $C'$ be the circuit obtained from it by replacing the leftmost and first in layer order $\swapgate^+$ box, labeled with $(i,j)$, with a $\swapgate$. Let $c_i$ and $c'_i$ be the initial inputs of circuits $C$ and $C'$ respectively. Then:
\begin{enumerate}
    \item $c_i$ differs from $c'_i$ in exactly one input. In particular, $c_j = c'_j\oplus c'_i$, and $c_k{=} c'_k$ for all $k{\neq }j$.
    \item At most $n{-}1$ boxes generate different linear functions between $C$ and $C'$. In particular, $box(k,j)$ generates $c_k\oplus c_j$ in $C$ and $c'_k \oplus c'_j = c_k\oplus c_i\oplus c_j$ in $C'$, if and only if $box(k,j)\preceq box(i,j)$
\end{enumerate}
\end{lemma}
 
We next show that we can always express the applications of Phase gates to linear functions generated by $C$ that are different from those generated by $C'$ using a constant number of Phase gates applied elsewhere by employing the identity \eq{7Term}.

\begin{lemma}\label{lem:replace}
Let $C$ and $C'$ be given as above, where $c_i$ and $c'_i$ denote input linear functions and $box(i, j)$, $i{<}j$ denotes the leftmost first in order $\swapgate^+$ box in $C$. For each Phase gate applied inside $C'$, there exists a corresponding schedule of either $1$ or $6$ Phase gates found in $C$ that computes this phase.  We give these schedules explicitly: 
\begin{enumerate}
    \item A $\pgate$ gate applied to the unique $c'_j$ in $C'$ such that $c'_j{\neq}c_j$ is equivalent to the $\pgate$ gate applied inside $box(i,j)$ in $C$.
    \item (``inversely'' to the above) A $\pgate$ gate applied inside $box(i,j)$ in $C'$ is equivalent to the $\pgate$ gate applied to $c_j$ in $C$.
    \item For $box(k,j){\prec} box(i,j)$, a $\pgate$ gate applied inside $box(k,j)$ in $C'$ is equivalent to a schedule of $\pgate^\dagger$ gates applied to/in $c_i,\,c_j,\,c_k,\,box(i,j),\,box(i,k),$ and $box(j,k)$ in $C$.
    \item A $\pgate$ gate applied to any other location in $C'$ is equivalent to the $\pgate$ applied to the same location in $C$.
\end{enumerate}
\end{lemma}
\begin{proof}
We begin by pointing out that $box(i,j)$, $box(j,k)$, and $box(i,k)$ are located on or before the layer $i$. This is clearly the case for $box(i,j)$ and $box(j,k)$ by construction. $box(i,k)\preceq box(i,j)$ as shown in \propos{order}. The linear functions generated by those boxes are the EXORs of two inputs.

First two cases follow from \lem{dif}. First, $c'_j = c_i\oplus c_j$, and $c_i\oplus c_j$ is generated by the $box(i,j)$ in $C$. The linear function generated by $box(i,j)$ in $C'$ is $c'_i\oplus c'_j = c_j$, and $c_j$ is an initial input to $C$. 

In case 3, $box(k,j)$ generates $c'_k\oplus c'_j$ in $C'$.  We can rewrite the application of the Phase gate to this linear function using the identity \eq{7Term} as follows:
\begin{equation*}
    c'_j\oplus c'_k = (c_i\oplus c_j)\oplus c_k = c_j{\oplus}c_k + c_i{\oplus}c_k + c_i{\oplus}c_j - c_i - c_j - c_k (\bmod 4)
\end{equation*}
The application of the $\pgate$ gates to the linear function in the circuit $c'$ on the right hand side of the equation is thus generated by the $\pgate$ and $\pgate^\dagger$ gates applied to the linear functions $c_i$, $c_j$, $c_k$, $box(i,j)$, $box(j,k)$, and $box(i,k)$ in the circuit $C$. We note in passing that this involves the introduction of the $P^\dagger$ gates, which can be rewritten as $P^\dagger = \pgate\zgate$ and Pauli-$\zgate$ can be moved to the side of a Clifford circuit, or better yet by allowing to multiply the identity in \eq{7Term} by modulo-4 integers, and thus allowing to treat arbitrary powers of the $\pgate$ gate.

Finally, case 4 also follows from \lem{dif}, since all other linear functions generated are equal in $C$ and $C'$.
\end{proof}

\subsection{Algorithm and Runtime Analysis} \label{ssec:Algo} 

\begin{algorithm}[h]
\caption{\texttt{InitializeS}}\label{alg:pcz}
\begin{algorithmic}
\State \textbf{Input:} -P-CZ-
\State \textbf{Initialize:} $S\gets n\times n$ zero matrix
\FOR{$P^k(i)\in$ -P-}
    \State $S[i,i]\gets k$
    \ENDFOR
\FOR{$CZ(i,j)\in$ -CZ-} 
     \State $(i,j)\gets (\min(i,j),\max(i,j))$
     \State $S[i,j]\gets 3$
     \State $S[i,i]\gets (S[i,i] + 1)\% 4$
     \State $S[j,j]\gets (S[j,j] + 1)\% 4$
    \ENDFOR
\State\textbf{Return:} $S$
\end{algorithmic}
\end{algorithm}

Finally, we discuss the algorithm that offers the construction of a depth-$5n$ implementation of arbitrary Hadamard-free Clifford circuit in time $O(n^3)$. For simplicity, we assume that the -P-CZ-CNOT- decomposition \cite{maslov2018shorter} and the depth-$5n$ implementation of the -CNOT- layer \cite{kutin2007computation} are given.  The algorithm starts by considering the $\swapgate^+$-free case of the northwest-triangular matrix diagonalization circuit and sets the Phases according to \lem{base}, see \alg{pcz}.  This part has complexity $O(n^2)$.  The Phase schedule is then updated iteratively, based on the proof of \thm{1}, as shown in the \alg{HfreeSynthesis}. This part has complexity $O(n^3)$, since adding one $\swapgate^+$ box requires updating at most $O(n)$ Phase schedules $S[i,j]$, with each update taking constant effort, and there can be at most $\frac{(n-1)n}{2}=O(n^2)$ $\swapgate^+$ boxes in the northwest-triangular matrix diagonalization circuit.  The overall complexity is thus $O(n^3)$.

\begin{algorithm}[h]
\caption{\texttt{FindPhaseSchedule}}\label{alg:HfreeSynthesis}
\begin{algorithmic}
\State \textbf{Input:} -P-CZ-CNOT-
\State $C\gets$ \texttt{InitializeC}(-CNOT-)
\State $S\gets$ \texttt{InitializeS}(-P-CZ-)

\State \texttt{\slash* Now, we iteratively add SWAP+ boxes in descending layer order *\slash}
\State \texttt{\slash* and update $S$ *\slash}

\FOR{$Box(i,j)\in C.\texttt{GetSWAP+}$} 
    \State \texttt{\slash* Case 1 \& 2 - Switch the $\pgate$ gates applied to box(i,j) and $c_j$ *\slash}
    \State $(S[j,j], S[i,j])\gets (S[i,j], S[j,j])$
    \FOR{$\{k=1$, $k\leq n$, $k\gets k+1\}$}
        \IF{$Box(k,j)\prec Box(i,j)$}{
            \State \texttt{\slash* Case 3 - Replace the $\pgate$ gates applied to box(k,j) with 6 $\pgate$ (or $\pgate^\dagger$) gates *\slash}
            \State $p_k\gets S[k,j]$
            \FOR{$(w_1,w_2)\in \{(i,i),(j,j),(k,k)\}$}    
                \State $S[w_1,w_2]\gets (S[w_1,w_2] + 3p_k)\%4$
                \ENDFOR
            \FOR{$(w_1,w_2)\in \{(i,j),(i,k)\}$}    
                \State $S[w_1,w_2]\gets (S[w_1,w_2] + p_k)\%4$
                \ENDFOR
        \ENDIF
        }
        \ENDFOR
    \ENDFOR
\State \textbf{Return:} $S$
 
\end{algorithmic}
\end{algorithm}

\subsection{Extension to Clifford circuits}\label{ssec:cliff}

The ability to implement -CNOT-CZ- circuits in depth $5n$ over the LNN architecture coupled with the ability to implement -CZ- circuits in depth $2n{+}2$ \cite{maslov2018shorter} directly implies the ability to implement arbitrary Clifford circuits in depth $5n + 2n{+}2 = 7n{+}2$ by relying on the decomposition -X-Z-P-CX-CZ-H-CZ-H-P- \cite[Lemma 8]{bravyi2021hadamard}. However, by considering local optimizations, one may reduce this figure to $7n{-}4$.

\begin{lemma}
Any Clifford transformation over $n$ qubits can be implemented over the LNN architecture as a circuit with the two-qubit gate depth of at most $7n{-}4$.   
\end{lemma}

\begin{proof}
Recall that the decomposition -X-Z-P$_1\widehat{\text{-CZ$_1$-}}$CX-H$_1\widehat{\text{-CZ$_2$-}}$H$_2$-P$_2$- \cite{bravyi2021hadamard} (here we changed the order of -CX- and -CZ-, which is allowed, mark all stages with subscripts, and use $\widehat{\;\;}$ to denote operation together with the qubit reversal) features the stage -H$_1$- with the Hadamard gates applying to all qubits. 

We first express $\widehat{\text{-CZ$_2$-}}$ as a depth $2n{+}2$ circuit according to \cite{maslov2018shorter}.  The layer -H$_1$- and the first four layers of the implementation of $\widehat{\text{-CZ$_2$-}}$ can be rewritten by pushing Hadamards through to the right and flipping the controls and targets of the $\cnotgate$ gates, as follows (illustrated for $n{=}7$), 
\begin{equation*}
\Qcircuit @C=0.7em @R=0.7em @!R {
& \qw & \gate{\hgate} & \ctrl{1}  & \qw       & \targ     & \qw       & \qw\\
& \qw & \gate{\hgate} & \targ     & \targ     & \ctrl{-1} & \ctrl{1}  & \qw\\
& \qw & \gate{\hgate} & \ctrl{1}  & \ctrl{-1} & \targ     & \targ     & \qw\\
& \qw & \gate{\hgate} & \targ     & \targ     & \ctrl{-1} & \ctrl{1}  & \qw\\
& \qw & \gate{\hgate} & \ctrl{1}  & \ctrl{-1} & \targ     & \targ     & \qw\\
& \qw & \gate{\hgate} & \targ     & \targ     & \ctrl{-1} & \ctrl{1}  & \qw\\
& \qw & \gate{\hgate} & \qw       & \ctrl{-1} & \qw       & \targ     & \qw\\}
\hspace{1em}\raisebox{-6em}{=}\hspace{1em}
\Qcircuit @C=0.7em @R=0.7em @!R {
& \qw & \targ     & \qw       & \ctrl{1}  & \qw       & \gate{\hgate} & \qw\\
& \qw & \ctrl{-1} & \ctrl{1}  & \targ     & \targ     & \gate{\hgate} & \qw\\
& \qw & \targ     & \targ     & \ctrl{1}  & \ctrl{-1} & \gate{\hgate} & \qw\\
& \qw & \ctrl{-1} & \ctrl{1}  & \targ     & \targ     & \gate{\hgate} & \qw\\
& \qw & \targ     & \targ     & \ctrl{1}  & \ctrl{-1} & \gate{\hgate} & \qw\\
& \qw & \ctrl{-1} & \ctrl{1}  & \targ     & \targ     & \gate{\hgate} & \qw\\
& \qw & \qw       & \targ     & \qw       & \ctrl{-1} & \gate{\hgate} & \qw\\}
\end{equation*}
This is justified since Phase gates need not be applied until after the first depth-$4$ $S$ stage \cite{maslov2018shorter}.  This means that the depth-$4$ $\cnotgate$ layer can be merged into the -CX- layer to its left, offering savings of $4$ layers in depth compared to the baseline construction of depth $7n{+}2$.

We next assume that $n$ is odd and show how to reduce the depth by additional $2$ layers.  To this end, focus on the last three two-qubit gate stages in the decomposition of $\widehat{\text{-CZ$_1$-}}$CX- into the depth $5n$ circuit, the layer -H$_1$-, and the first two-qubit gate stage in the now-reduced decomposition of $\widehat{\text{-CZ$_2$-}}$.  Note that the gates in these four layers apply only to pairs of qubits $(1,2)$, $(3,4)$, ... . For each pair of qubits $(k, k{+}1)$, the last three two-qubit gate stages in the decomposition of $\widehat{\text{-CZ$_1$-}}$CX- are either SWAP or SWAP+, with Phase gate or without it; the -H$_1$- is always the layer of two Hadamards, and the first two-qubit gate stage of $\widehat{\text{-CZ$_2$-}}$ is a set of $\pgate$ gates followed by the $\cnotgate(k,k{+}1)$.  We group these cases into the following three and show that each can be implemented as a circuit of depth at most $2$ (here, $a,b,c,d,$ and $e$ take Boolean values to indicate if the given gate is present, $1$, on not, $0$).

\paragraph{Case 1: Last box is \text{SWAP}.} 
\begin{equation*}
\Qcircuit @C=1.0em @R=0.2em @!R { \\
	 	 & \ctrl{1} & \qw & \targ & \ctrl{1} & \gate{\mathrm{P}^b} & \gate{\mathrm{H}} & \gate{\mathrm{P}^d} & \ctrl{1} & \qw & \qw\\
	 	 & \targ & \gate{\mathrm{P}^a} & \ctrl{-1} & \targ & \gate{\mathrm{P}^c} & \gate{\mathrm{H}} & \gate{\mathrm{P}^e} & \targ & \qw & \qw\\}
\hspace{1em}\raisebox{-2.5em}{=}\hspace{1em}
\Qcircuit @C=1.0em @R=0.2em @!R { \\
	 	 & \gate{\mathrm{P}^c} & \targ & \gate{\mathrm{P}^a} & \gate{\mathrm{H}} & \gate{\mathrm{P}^e} & \targ & \gate{\mathrm{P}^d} & \qw & \qw\\
	 	& \gate{\mathrm{P}^b} & \ctrl{-1} & \qw & \gate{\mathrm{H}}  & \qw & \ctrl{-1} & \qw & \qw & \qw\\}
\end{equation*}

\paragraph{Case 2: Last box is \text{SWAP+} and $e{=}0$.}
\begin{equation*}
\Qcircuit @C=1.0em @R=0.2em @!R { \\
	 	 & \ctrl{1} & \qw & \targ & \gate{\mathrm{P}^b} & \gate{\mathrm{H}} & \gate{\mathrm{P}^d} & \ctrl{1} & \qw & \qw\\
	 	 & \targ & \gate{\mathrm{P}^a} & \ctrl{-1}  & \gate{\mathrm{P}^c} & \gate{\mathrm{H}} & \qw & \targ & \qw & \qw\\}
\hspace{1em}\raisebox{-2.5em}{=}\hspace{1em}
\Qcircuit @C=1.0em @R=0.2em @!R { \\
	 	 & \qw & \ctrl{1} & \qw & \gate{\mathrm{H}} & \gate{\mathrm{P}^d} &\qw\\
	 	& \gate{\mathrm{P}^b} & \targ & \gate{\mathrm{P}^{a+c}} & \gate{\mathrm{H}}  & \qw &\qw\\}
\end{equation*}

\paragraph{Case 3: Last box is \text{SWAP+} and $e{=}1$.}
\begin{equation*}
\Qcircuit @C=1.0em @R=0.2em @!R { \\
	 	 & \ctrl{1} & \qw & \targ & \gate{\mathrm{P}^b} & \gate{\mathrm{H}} & \gate{\mathrm{P}^d} & \ctrl{1} & \qw & \qw\\
	 	 & \targ & \gate{\mathrm{P}^a} & \ctrl{-1}  & \gate{\mathrm{P}^c} & \gate{\mathrm{H}} & \gate{\mathrm{P}} & \targ & \qw & \qw\\}
\hspace{1em}\raisebox{-2.5em}{=}\hspace{1em}
\Qcircuit @C=1.0em @R=0.2em @!R { \\
	 	 & \qw & \ctrl{1}  & \gate{\mathrm{H}} &\ctrl{1}&\qw& \gate{\mathrm{P}^{d+1}} &\qw\\
	 	 & \gate{\mathrm{P}^b} & \targ & \gate{\mathrm{P}^{a+c}} & \targ&\gate{\mathrm{H}}  & \gate{\mathrm{P}} &\qw\\}
\end{equation*}

This means that for odd $n$ a Clifford circuit can be implemented in depth $7n{-}4 = 7n{+}2 - 4 - 2$.  The case of even $n$ is handled similarly, except to enable the gate cancellations in the routine reducing the depth by $2$, we implement the inverse of -CX- and list the gates in inverted order.
\end{proof}

\section{-CNOT-CZ- circuits over all-to-all architecture} \label{sec:all-to-all}

We next turn our attention to the all-to-all connected architecture. There are a number of methods available to synthesize a $\cnotgate$ circuit, targeting various optimization criteria. Here, we consider two methods: MZ synthesis method that achieves the circuit depth $n+O(\log^2n)$ \cite{maslov2022depth}, which is currently the best known method for the number of qubits up to $1{,}345{,}000$, and the PMH method achieving asymptotically optimal gate count of $O(n^2/\log(n))$ \cite{patel2008optimal}. 

Our goal is to establish a $\czgate$ basis in the $\cnotgate$ circuits found by these algorithms, and when this is impossible, find the smallest number of $\czgate$ gate insertions that allows generating a $\czgate$ basis. We note that the basis must contain $\frac{(n-1)n}{2}$ functions, and thus for the algorithms showing sub-quadratic scaling such as the PMH \cite{patel2008optimal}, this will eventually become impossible without the addition of $\czgate$ gates.  However, given the high leading constant \cite{de2021reducing} and our experiments in \fig{gateOptimal}, we do not anticipate this to happen for reasonably small numbers $n$.  We also note that due to the input-dependant structure of the circuits generated by MZ and PMH algorithms, it is difficult to come up with proofs of performance guarantee.  We thus resort to the combination of heuristics and random tests.

\subsection{A Heuristic Algorithm} \label{subsec:HeuristicAlgo}

\begin{algorithm}[h]
\caption{\texttt{InsertCZ}}\label{alg:insertCZ}
\begin{algorithmic}
\State \textbf{Input:} -CNOT-
\State $Z\gets$ $\czgate$ vectors generated by the -CNOT- layer
\State $B\gets$ basis for $Z$
\State $B_\perp \gets$ basis for the orthogonal complement of $B$
\FOR{$\czgate(x_i,x_j)\in B_\perp$} 
    \FOR{$(i,j)$ s.t. $x_i\in l_i, x_j \in l_j$ and wires $i,j$ are unoccupied in layer $l$}
    \State $v \gets$ $\czgate$ vector corresponding to $l_i\oplus l_j$
    \State \texttt{\slash* Checking whether $v$ is spanned by $B$ *\slash}
    \FOR{$b\in B$}:
        \IF{$b\cdot v =1$}{
            \State $v\gets v\oplus b$
            \ENDIF}
        \ENDFOR
        \IF{$v\neq 0$}{
            \State $l$.add($\czgate(i,j)$)
            \State $B.$add($v$)
            \IF{${v}_{i,j} = 1$}{
                \State \texttt{\slash* The added CZ generates $\czgate(x_i,x_j)$ as desired *\slash}
                \State \textbf{Continue}
            \ELSE{
                \State \texttt{\slash* The added CZ generates a different missing basis vector in $B_\perp$ *\slash}
                \State $B_\perp.$remove($v$)
            \ENDIF}}
            }\ENDIF
        \ENDFOR
    \State \texttt{\slash* $\czgate(x_i,x_j)$ cannot be inserted without increasing depth *\slash}
    \State -CNOT-.add($\czgate(x_i,x_j)$)
    \ENDFOR
\State\textbf{Return:} -CNOT-
\end{algorithmic}
\end{algorithm}

In case there are insufficient linear functions in a given $\cnotgate$ circuit to form a $\czgate$ basis, we can seek to conditionally add $\czgate$ gates to generate new linearly independent vectors in the space $\mathbb{F}_2^{(n-1)n/2}$. The algorithm accomplishing it works as follows. First, find the $\czgate$ subspace spanned by the linear functions generated, with a basis $B$, and its orthogonal complement, with the basis $B^\perp$. $dim(B^\perp)$ specifies how many linearly independent conditional $\czgate$ gates need to be added since each linearly independent $\czgate$ gate reduces the dimensionality of $B^\perp$ by one. We obtain $B$ and $B^\perp$ by performing Gaussian elimination on the $\czgate$ vectors generated.  These observations allow establishing the exact minimum number of the $\czgate$ gates that need to be added and computing them, thus addressing the scenario of the gate count minimization.  We next consider depth minimization.

For each layer in the given $\cnotgate$ circuit, we greedily add the $\czgate$ to it, so long as it applies to two unoccupied qubits and reduces the dimensionality of $B^\perp$. If no such opportunity exists, we create a new layer at the end of the circuit and add the necessary $\czgate$ to it.

The runtime of the algorithm depends on two steps: performing Gaussian elimination, and checking whether $l_i\oplus l_j$, the EXOR of linear functions on a pair of unoccupied qubits in a given layer, is spanned by $B$. Gaussian elimination algorithm takes $O(n^6)$ time, since CZ vectors have dimension $(n-1)n/2 = O(n^2)$ and the number of linear functions generated is $O(n^2)$. It takes time $O(n^4)$ to check whether each $l_i\oplus l_j$ is spanned by $B$; however, the number of times this check is invoked depends on the number of missing basis vectors as well as the available spots to go through, which can be $O(n^4)$ in the worst case. 


Given the performance of the above simple heuristic, we found it to be not worth the time to investigate other approaches to finding a full basis.


\subsection{Algorithm Performance on Two Synthesis Methods}\label{subsec:HeuristicResults}

In this section, we test the \alg{insertCZ} on random $\cnotgate$ circuits offered by the two synthesis algorithms considered, MZ \cite{maslov2022depth} and PMH \cite{patel2008optimal}.

We sample linear reversible transformations uniformly at random using the method given by \cite[Alg. 4]{bravyi2021hadamard}.  We next apply MZ and PMH algorithms to generate two $\cnotgate$ circuits.  We check whether the $\cnotgate$ circuits generate enough linear functions to form a basis in the $\czgate$ space via Gaussian elimination; if not, we apply the heuristic algorithm that conditionally inserts the $\czgate$ gates on the available qubits to generate the remaining linear functions.  Finally, we compare the circuit depth and gate count before and after the $\czgate$ gate insertion.

The results for the depth-optimized MZ synthesis algorithm \cite{maslov2022depth} are summarized in \fig{depthOptimal}. Remarkably, we found that the $\cnotgate$ circuits synthesized generate a $\czgate$ basis in about $30\%$ of cases when the number of qubits is between $25$ and $125$.  When there are basis vectors missing, the insertion of a linearly independent $\czgate$ gate via \alg{insertCZ} most often comes at no increase in the total depth. The average increase in depth is less than $1$ for each number of qubits $n$ tested over 10,000 trials.

For PMH algorithm \cite{patel2008optimal}, we focus on the minimization of the gate count.  The results are summarized in \fig{gateOptimal}.  We found that for circuits spanning $45$+ qubits, the linear functions generated always form a basis in the $\czgate$ circuit space over $1000$ random linear reversible operations tried for each number of qubits $n$ up to $100$. 

\begin{figure}[t]
  \centering
  \includegraphics[width=.6\linewidth]{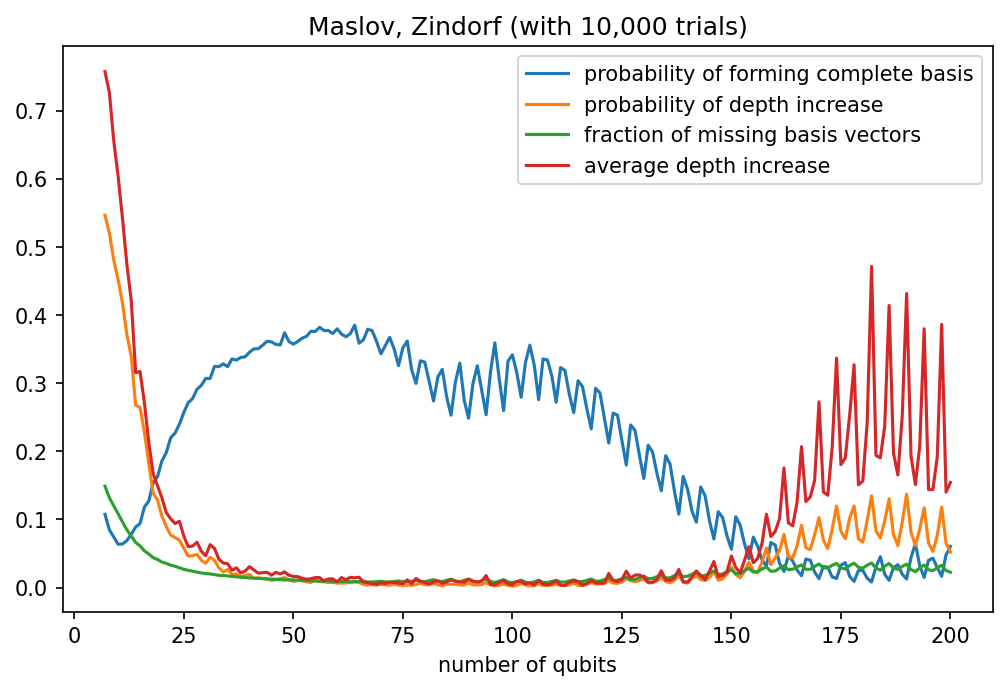}
  \caption{MZ algorithm \cite{maslov2022depth}: the probability that MZ implementation of a random linear reversible function contains a $\czgate$ basis (blue), the probability the depth increases (orange), the fraction of missing basis vectors (green), and the average increase in depth (red) as a function of the number of qubits.
  }
\label{fig:depthOptimal}
\end{figure}

\begin{figure}[t]
\begin{subfigure}{.485\textwidth}
  \centering
  \includegraphics[width=.9\linewidth]{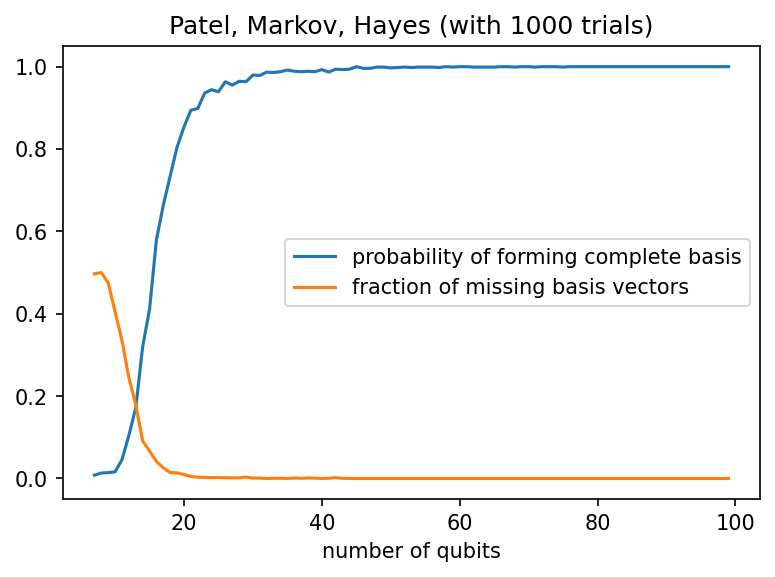}
  \caption{Proportion of the $\czgate$ vectors generated as a function of the number of qubits.}
\end{subfigure}%
\begin{subfigure}{.5\textwidth}
  \centering
  \includegraphics[width=.9\linewidth]{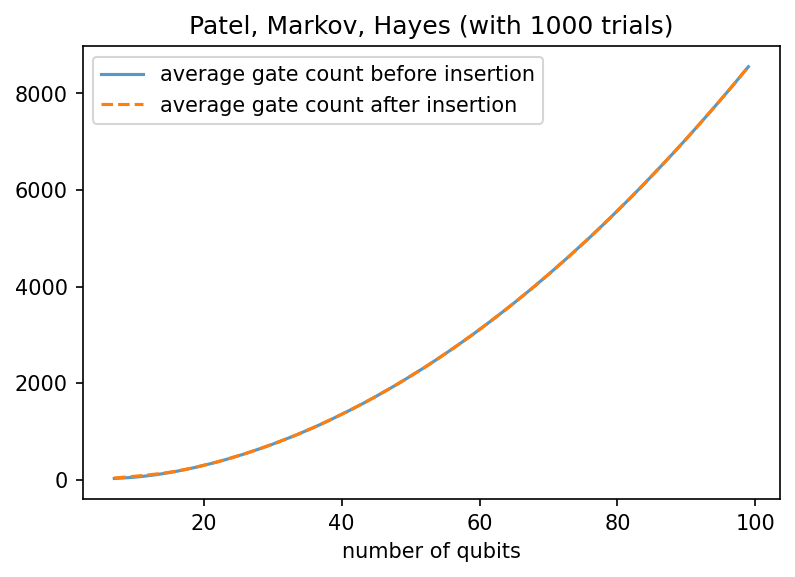}
  \caption{Gate counts before and after $\czgate$ gate insertions as a function of the number of qubits.}
\end{subfigure}
\caption{Summary of results for the PMH algorithm \cite{patel2008optimal}.}
\label{fig:gateOptimal}
\end{figure}

\section{Discussion and Conclusion}\label{sec:conc}

In this paper, we studied the ability to insert Phase gates into $\cnotgate$ circuits to induce arbitrary Hadamard-free Clifford transformation generated by the given linear reversible function.  The reason this is often possible to accomplish is the application of Phase gates to linear functions computed inside the $\cnotgate$ circuit generates vectors in the linear space spanning $\czgate$ circuits, and if sufficiently many independent vectors are found that form the basis in the $\czgate$ circuits space, this guarantees the ability to implement any $\czgate$ transform.  In particular, we showed a depth upper bound of $5n$ for the implementation of arbitrary Hadamard-free Clifford transformation, matching the best-known upper bound of $5n$ to generate arbitrary linear reversible transformation over LNN architecture.  We also developed a heuristic that offers an average depth increase of less than $1$ for all number of inputs $n$ tried ($n {\leq} 200$) to turn a linear reversible circuit into a Hadamard-free Clifford circuit generated by it over all-to-all architecture.  In other words, depths of best-known Hadamard-free Clifford and linear reversible computations appear to be essentially equal in at least two scenarios---LNN and all-to-all architectures.  This is surprising, given the difference between group sizes---$2^{n^2+O(1)}$ for linear reversible computations vs $2^{1.5n^2+O(1)}$ for Hadamard-free Clifford circuits.  Indeed, one would expect the bigger group to demand higher circuit complexity.

We would like to highlight the important role the sorting network (also known as the qubit reversal circuit) composed with $\swapgate$ gates plays over the LNN architecture.  A modification of the sorting network that removes a subset of the $\cnotgate$ gates (each $\swapgate$ can be implemented by three $\cnotgate$ gates) and inserts Phase gates implements arbitrary -CZ- transformation up to qubit reversal in depth $2n{+}2$ \cite{maslov2018shorter}. A modification of the sorting network that erases a subset of the $\cnotgate$s implements arbitrary qubit swapping circuit in depth no more than $3n$ \cite{kutin2007computation}. A modification of the set of two back-to-back sorting networks that removes a subset of the $\cnotgate$ gates implements arbitrary linear reversible function in depth no more than $5n$ \cite{kutin2007computation}.  A modification of the set of two back-to-back sorting networks that removes a subset of the $\cnotgate$ gates and inserts Phase gates implements arbitrary Hadamard-free Clifford circuit in depth no more than $5n$.  Finally, a modification of the set of three back-to-back sorting networks with a Hadamard layer between some two of them that removes a subset of the $\cnotgate$ gates and inserts Phase gates, together with a few local optimizations, implements arbitrary Clifford circuit in depth no more than $7n{-}4$.

\section*{Acknowledgement}

We would like to thank Ben Zindorf (University College London) for his helpful discussions during the early stage of this research.

\end{document}